\documentclass[12pt]{article}
\usepackage{amsmath,amsthm}
\usepackage{amssymb,amsfonts}
\usepackage[english]{babel}
\usepackage{gastex}
\usepackage{graphicx}
\usepackage{cite}
\usepackage[margin=3cm]{geometry}
\DeclareMathOperator{\per}{per}
\newtheorem{Thm}{Theorem}
\newtheorem{Rem}{Remark}
\newtheorem{Prop}{Proposition}
\newtheorem{Cor}{Corollary}
\newtheorem{Lem}{Lemma}

\date{}
\title{On ternary square-free circular words}
\author{Arseny M. Shur\\
\small Ural State University\\[-0.8ex]
\small Ekaterinburg, Russia\\
\small \texttt{arseny.shur@usu.ru}\\}

\begin{document}
\maketitle

\begin{abstract} 
Circular words are cyclically ordered finite sequences of letters. We give a computer-free proof of the following result by Currie: square-free circular words over the ternary alphabet exist for all lengths $l$ except for 5, 7, 9, 10, 14, and 17.
Our proof reveals an interesting connection between ternary square-free circular words and closed walks in the $K_{3{,}3}$ graph. In addition, our proof implies an exponential lower bound on the number of such circular words of length $l$ and allows one to list all lengths $l$ for which such a circular word is unique up to isomorphism.
\end{abstract}

\section*{Introduction}

Circular words form an obvious variation of ordinary words: we just link up the ends of a word, thus getting a cyclic sequence of letters without begin or end. Circular words arise naturally as the labels of closed walks in graphs and automata.

Distinctive properties of circular words are mostly due to the following two features. First, the distance between two positions in a circular word can be counted in any of the two directions (``clockwise'' and ``counterclockwise''; or ``the distance from $i$ to $j$'' and ``the distance from $j$ to $i$''). Second, the factors of a circular word are ordinary words. The \textit{primal} factors are obtained just by cutting the circular word in some point. All primal factors of a circular word are cyclic shifts (or \textit{conjugates}) of each other. Thus, a circular word can be viewed as the representation of a conjugacy class of ordinary words.

Most properties of words can be expressed in terms of possessing or avoiding some regularities. For example, ``to be a square'' is a possession property, ``to be square-free'' is an avoidance property, and ``to be a minimal square'' is a sort of intersection of the first two properties. The study of avoidance properties of circular words is closely connected to the study of squares. In particular, Proposition~\ref{l1} below binds minimal squares and square-free circular words. 

The problem we study here has a graph theory origin. Namely, in \cite{AGHR} the existence of non-repetitive walks in coloured graphs was studied. The authors conjectured\footnote{Currie \cite{Cur} mentioned that this conjecture is originally due to R. J. Simpson but provided no reference. More details on non-repetitive colourings of graphs can be found in \cite{Gry}.} which cycles $C_l$ can be 3-coloured such that the label of any path is square-free. The existence of such a colouring is obviously equivalent to the existence of a circular square-free word of length $l$ over the ternary alphabet. This conjecture was proved by Currie \cite{Cur}:

\begin{Thm} \label{sf}
Square-free circular words over the ternary alphabet exist for all lengths $l$ except for $l=5, 7, 9, 10, 14, 17$.
\end{Thm}

The proof given by Currie consists of a construction allowing one to obtain a square-free circular word of any length $l\ge180$ and a computer search covering the cases $1\le l\le179$. In particular, this proof gives very little information about the structure of ternary square-free circular words and tells nothing about the nature of the exceptions found.

The aim of this paper is to give a computer-free proof of Theorem~\ref{sf}. First we reduce the problem for ternary circular words to the problem for their binary \textit{codewords}. Then we define a weight function on the $K_{3{,}3}$ graph in order to express the main feature of these codewords. Roughly speaking, a closed walk in the obtained graph, having weight $l$ and possessing some additional property, generates a square-free ternary circular word of length $l$. Finally we describe a way to construct such walks of any given weight $l\ge18$.

\section{Preliminaries}

We recall some notation and definitions on words, see \cite{Lo} for more background.

\smallskip
An \textit{alphabet} $\Sigma$ is a nonempty finite set, the elements of which are called \textit{letters}. \textit{Words} are finite sequences of letters. As usual, we write $\Sigma^*$ [$\Sigma^+$] for the set of all words over $\Sigma$, including [respectively, excluding] the \textit{empty word} $\lambda$. A word $u$ is a {\it factor} [{\it prefix}, {\it suffix}] of a word $w$ if $w$ can be represented as $\bar{v}u\hat{v}$ [respectively, $u\hat{v}$, $\bar{v}u$] for some (possibly empty) words $\bar{v}$ and $\hat{v}$. A {\it factor} [{\it prefix}, {\it suffix}] of $w$ is called \textit{proper} if it does not coincide with $w$. Words $u$ and $w$ are \textit{conjugates} if $u=\bar{v}\hat{v}$, $w=\hat{v}\bar{v}$ for some words $\bar{v}$ and $\hat{v}$. Conjugacy is obviously an equivalence relation. Words $u$ and $v$ are \textit{isomorphic} if $u=\sigma(v)$ for some bijection $\sigma$ of the alphabet(s).

A word $w\in\Sigma^*$ can be viewed as a function $\{1,\ldots,n\}\to\Sigma$. Then a \textit{period} of $w$ is any period of this function. The \textit{exponent} of $w$ is given by $\exp(w)=|w|/\per(w)$, where $\per(w)$ is the minimal period of $w$, $|w|$ is the length of $w$. The word $w$ is said to be $\beta$-\textit{free} if all its factors have exponents less than $\beta$. The prefix of $w$ of length $\per(w)$ is called the \textit{root} of $w$. A word of exponent 2 is a \textit{square}. The square is \textit{minimal} if it contains no other squares as factors. 

We use boldface letters to denote \textit{circular words}, which are finite \textit{cyclic} sequences of letters. In addition, we write $(w)$ for the circular word obtained by linking up the ends of the word $w$.


There is an obvious one-to-one correspondence between circular words and conjugacy classes of ordinary words. Thus, the definition of $\beta$-freeness (in particular, square-freeness) naturally extends to circular words. Square-free circular words are closely connected to minimal squares:

\begin{Prop} \label{l1}
The word $u^2$ is a minimal square if and only if the circular word $(u)$ is square-free. 
\end{Prop}

\begin{proof}
Necessity is trivial, because $u^2$ contains all conjugates of $u$ as factors. To prove sufficiency by contradiction, assume that the circular word $(u)$ is square-free but $u^2$ is not a minimal square. Hence, $u^2$ contains a proper factor $v^2$. If $|v^2|\le|u|$, then $v^2$ is a factor of a conjugate of $u$, contradicting to the square-freeness of $(u)$. Now examine the case $|v^2|>|u|$. Then $uu=xvvy$ and without loss of generality $|x|\ge|y|$. So, we have $v=v_1v_2$, $u=xv_1=v_2v_1v_2y$. Moreover, since $|x|{+}|y|<|u|$, the prefix $v_2v_1v_2$ and the suffix $v_1$ overlap in $u$. The occurrences of $v_1$ cannot overlap, because $u$ is square-free. Thus, the suffix $v_1$ intersects only with the second occurrence of $v_2$ in the considered prefix. Hence, some word $z$ is a prefix of $v_1$ and simultaneously a suffix of $v_2$. Then the prefix $v_2v_1$ of $u$ contains $z^2$, which is impossible. 
\end{proof}

\section{Proof of the main result}

We prove Theorem~\ref{sf} in three main steps described in Sections~\ref{pan}--\ref{walks}. We introduce binary codewords and their basic properties in Section~\ref{pan}. As a corollary, we prove the statement of Theorem~\ref{sf} for $l\le8$. In Section~\ref{k_33} we represent these codewords by the walks in the weighted $K_{3{,}3}$ graph and claim Theorem~\ref{sf} for $9\le l\le21$. Finally, in Section~\ref{walks} we describe a way to construct the walks of given weight and prove Theorem~\ref{sf} for $l\ge22$.

\subsection{Reduction to binary words: Pansiot's encoding}\label{pan}

Any ternary square-free word $w$ of length $l\ge3$ can be encoded, up to isomorphism, by a codeword $\hat{w}\in\{0,1\}^+$ of length $l{-}2$. The encoding rule is as follows: 
$$
\hat{w}(i)=\begin{cases}
0,& \text{if } w(i{+}2)=w(i),\\
1,& \text{otherwise}.
\end{cases}
$$
For example,
$$
\arraycolsep=1pt
\begin{array}{rclllllllll}
w&=&a&b&c&b&a&c&b&c&\ldots\\
\hat{w}&=&1&0&1&1&1&0&&&\ldots\\
\end{array}
$$
This type of encoding for an arbitrary $k$-ary $\frac{k{-}1}{k{-}2}$-free word was suggested by Pansiot \cite{Pan} as a tool to operate with famous Dejean's conjecture on avoidable exponents (see \cite{Dej}). Pansiot's encoding can be easily extended to circular words: any ternary square-free circular word $\mathbf{w}$ can be encoded, up to isomorphism, by a binary circular codeword $\hat{\mathbf{w}}$ of the same length $l\ge3$, as in the following example:

\centerline{
\unitlength=0.8mm
\begin{picture}(60,27)(0,-3)
\gasset{Nw=0,Nh=0,Nframe=n}
\node(f0)(0,10){$a$}
\node(f1)(2.9,17.1){$b$}
\node(f2)(10,20){$c$}
\node(f3)(17.1,17.1){$b$}
\node(f4)(20,10){$a$}
\node(f5)(17.1,2.9){$c$}
\node(f6)(10,0){$b$}
\node(f7)(2.9,2.9){$c$}
\node(o1)(5,10){}
\node(o2)(10,15){}
\drawedge[curvedepth=2](o1,o2){}
\put(30,10){\makebox(0,0)[cb]{$\longrightarrow$}}
\node(c0)(40,10){$1$}
\node(c1)(42.9,17.1){$0$}
\node(c2)(50,20){$1$}
\node(c3)(57.1,17.1){$1$}
\node(c4)(60,10){$1$}
\node(c5)(57.1,2.9){$0$}
\node(c6)(50,0){$1$}
\node(c7)(42.9,2.9){$1$}
\node(p1)(45,10){}
\node(p2)(50,15){}
\drawedge[curvedepth=2](p1,p2){}
\end{picture} }

For convenience, we always write circular codewords starting with 0. Note that a word $w\in\{a,b,c\}^+$ of length $l\ge3$ admits Pansiot's encoding if and only if $w$ contains no squares of letters, and every word $u\in\{0,1\}^+$ is a codeword of such a word. The former statement (but not the latter!) remains true for circular words over $\{a,b,c\}$ and circular codewords of length $l\ge3$ over $\{0,1\}$. To prove Theorem~\ref{sf}, we produce for each $l\ge3$ a codeword of length $l$ which encodes a square-free circular word or show that no such codeword exists. Let us say that a codeword (circular or not) is \textit{square-free} if it encodes a square-free word. We start with a trivial but important observation:
\begin{itemize}
\item[($*$)] a circular codeword of length $l$ is square-free if and only if no one of its factors of length $\le l{-}2$ encodes a minimal square. 
\end{itemize}

\begin{Rem}\label{r123}
Circular words $(a)$, $(ab)$, and $(abc)$ are unique, up to isomorphism, square-free circular words of length 1, 2, and 3, respectively.
\end{Rem}

\begin{Lem}\label{l_le8}
A square-free circular codeword of length $l$ such that $4\le l\le8$ coincides with one of the words ${\mathbf c}_4=(0101)$, ${\mathbf c}_6=(011011)$, or ${\mathbf c}_8=(01110111)$.
\end{Lem}

\begin{Cor}\label{c_le8}
Theorem \ref{sf} holds for $l\le8$.
\end{Cor}

\begin{proof}[Proof of Lemma~\ref{l_le8}]
A square-free circular word of length 4 has two occurrences of some letter, say $a$. Since neither of the two distances between these occurrences equals 1, both distances are equal to 2. Since the word $(abab)$ is not square-free, $(abac)$ is a unique, up to isomorphism, square-free circular word of length 4. Its codeword is ${\mathbf c}_4$. 

We make a few auxiliary notes before studying the lengths $l\ge5$. To use ($*$), we find the codewords of short minimal squares. By Proposition~\ref{l1}, the roots of such squares are the primal factors of square-free circular words. All squares are considered up to isomorphism. Since the encoded word cannot contain squares of period 1, we start with the period 2. By Remark~\ref{r123}, the only minimal square of period 2 is $abab$; it is encoded to $00$. Similarly, the only minimal square of period 3 is $abcabc$, and its codeword is $1111$. By ($*$), 
\begin{itemize}
\item[($**$)] square-free codewords of length $l\ge4$ [$l\ge6$] have no factors $00$ [respectively, $1111$].
\end{itemize}
The two minimal squares of period 4 are $abacabac$, coded by $010101$, and $abcbabcb$, coded by $101010$. Since any $0$ in a circular codeword is surrounded by two $1$'s, such a codeword of length $l\ge8$ avoids $010101$ and $101010$ if and only if it avoids $01010$.

Note that a ternary circular word $\mathbf{u}$ of length $5$ has at least two pairs of equal letters. For any such pair, one of the distances is either 1 (then $\mathbf{u}$ contains the square of a letter) or 2. Thus, the codeword of $\mathbf{u}$, if exists, has at least two $0$'s. If one of the distances between two $0$'s in this codeword is $2$, then $\mathbf{u}(i)=\mathbf{u}(i{+}2)$ and $\mathbf{u}(i{+}2)=\mathbf{u}(i{+}4)$ for some $i$ (the addition is modulo 5). Since $i{+}4=i{-}1 \pmod 5$, $\mathbf{u}$ contains the square of a letter and then has no codeword, a contradiction. Hence, the codeword of $\mathbf{u}$ has the factor $00$, implying that $\mathbf{u}$ is not square-free by ($**$).

In view of ($**$), the circular words (010101), (010111), and (011011) are the only candidates for square-free circular codewords of length 6, while (0101011) and (0110111) are the only such candidates of length 7. Square-free circular codewords of length 8 have no factor 01010, so the candidates are (01011011) and (01110111). It can be directly verified that all candidates apart from the words ${\mathbf c}_6$ and ${\mathbf c}_8$ fail. The lemma is proved.
\end{proof}

\begin{Rem}\label{r68}
The minimal squares of period 6 have the codewords $0110110110$, $1101101101$, and $1011011011$. For circular words, the avoidance of these words can be reduced to the avoidance of $011011011$ and $110110110$. Similarly, the avoidance of the codewords of minimal squares of period 8 is reduced to the avoidance of the word $11101110111$. 
\end{Rem}

\subsection{Reduction to the walks in the $K_{3{,}3}$ graph}\label{k_33}

Note that $0$'s in the codeword correspond to the ``jumps'' of one letter over another letter in the encoded word. There are six such jumps, represented by the factors $aba$, $bcb$, $cac$, $aca$, $bab$, and $cbc$. We call the first three jumps \textit{right} and the remaining jumps \textit{left}. It is easy to see that a right jump in a square-free circular word is followed by the left jump, and vice versa. Thus, the number of $0$'s in any square-free circular codeword is even. 

Which of the three left [right] jumps follows the given right [respectively, left] jump? It depends on the number of $1$'s which separate the $0$'s responsible for the jumps. Namely, the next jump is obtained from the previous one by
\begin{itemize}
\item[--] changing the central letter (e.\,g., $aba\leftrightarrow aca$) if the $0$'s are separated by $1$;
\item[--] changing the side letters (e.\,g., $aba\leftrightarrow cbc$) if the $0$'s are separated by $11$;
\item[--] switching the roles of the letters (e.\,g., $aba\leftrightarrow bab$) if the $0$'s are separated by $111$.
\end{itemize}
In order to describe square-free codewords we use the complete bipartite graph $K_{3,3}$, see Fig.~\ref{k33}. The jumps, partitioned into right and left, are represented by the vertices. The number of $1$'s needed to encode the sequence of two jumps equals the weight of the edge connecting the corresponding vertices. 

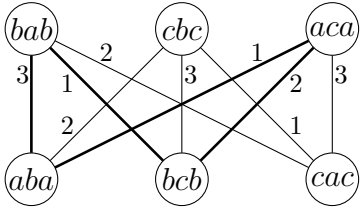
\begin{figure}[htb] 
\centerline{
\begin{picture}(60,29)(0,1)
\gasset{Nw=7,Nh=7,AHnb=0,ELdist=0.3,ELpos=70}
\node(f1)(10,5){$aba$}
\node(f2)(30,5){$bcb$}
\node(f3)(50,5){$cac$}
\node(c1)(10,25){$bab$}
\node(c2)(30,25){$cbc$}
\node(c3)(50,25){$aca$}
\drawedge[linewidth=0.35](f1,c1){\footnotesize 3}
\drawedge[ELpos=30](f1,c2){\footnotesize 2}
\drawedge[ELpos=77,linewidth=0.35](f1,c3){\footnotesize 1}
\drawedge[linewidth=0.35](f2,c1){\footnotesize 1}
\drawedge[ELside=r](f2,c2){\footnotesize 3}
\drawedge[ELside=r,linewidth=0.35](f2,c3){\footnotesize 2}
\drawedge[ELside=r,ELpos=77](f3,c1){\footnotesize 2}
\drawedge[ELside=r,ELpos=30](f3,c2){\footnotesize 1}
\drawedge[ELside=r](f3,c3){\footnotesize 3}
\end{picture} }
\caption{\footnotesize\sl The graph of jumps in ternary circular words. The closed walk marked in boldface represents the square-free codeword $(01011010111)$ which encodes the circular word $(abacabcbabc)$.} \label{k33}
\end{figure}

\begin{Rem} \label{r1}
Each square-free circular codeword of length at least 4 corresponds to a closed walk in the weighted graph shown in Fig.~\ref{k33}. Conversely, each closed walk in this graph defines a circular codeword, but in general these codewords are not necessarily square-free. 
\end{Rem}

So, to prove the statement of Theorem \ref{sf} for a given $l$ we either exhibit a closed walk which defines a square-free circular codeword of length $l$, or prove that no such walk exists.

\medskip
A walk in the obtained graph is uniquely determined by the initial vertex and the sequence of edge weights. Due to symmetry, this sequence of weights determines whether the walk is closed independently of the initial vertex. Since we are interested in the codewords rather than the encoded words, we consider the walks just as sequences of weights (i.\,e., words over $\{1,2,3\}$). For example, the closed walks $11$, $22$, and $33$ define the codewords ${\mathbf c}_4$, ${\mathbf c}_6$, and ${\mathbf c}_8$ respectively. For convenience, we define the weight of a closed walk to be the total weight of its edges plus its length. Then, the weight of a closed walk equals the length of the circular codeword defined by this walk.

Any closed walk is a combination of simple cycles (a closed walk of length two is considered as a simple cycle also). Any simple cycle is defined by a circular word over $\{1,2,3\}$. We enumerate all simple cycles together with the corresponding circular codewords and their lengths:
\begin{equation}
\begin{array}{llr}
(11)&(0101)={\mathbf c}_4&4\\
(22)&(011011)={\mathbf c}_6&6\\
(33)&(01110111)={\mathbf c}_8&8\\
(1213)&(01011010111)&11\\
(1232)&(010110111011)&12\\
(1323)&(0101110110111)&13\\
(121212)&(010110101101011)&15\\
(123123)&(010110111010110111)&18\\
(132132)&(010111011010111011)&18\\
(131313)&(010111010111010111)&18\\
(232323)&(011011101101110110111)&21
\end{array} \label{tab1}
\end{equation}

In what follows, we write ${\mathbf u}_1$+\ldots+${\mathbf u}_n$ for the set of closed walks obtained as a combination of simple cycles ${\mathbf u}_1, \ldots,{\mathbf u}_n$.

\begin{Lem}\label{cycl}
All the codewords in (\ref{tab1}) are square-free.
\end{Lem}
\begin{proof}
The codewords ${\mathbf c}_4$, ${\mathbf c}_6$, and ${\mathbf c}_8$ are square-free by Lemma \ref{l_le8}. Note that there exist no closed walks generating a circular codeword of length 9. Further, a circular codeword of length 10 is generated by a walk from the set (11)+(22). Such a codeword is not square-free, because it certainly contains the factor $01010$. By Proposition~\ref{l1}, no minimal squares of period 9 or 10 exist. The codewords of minimal squares of period $p\le8$ (and hence, of period $p\le10$) are calculated in Section~\ref{pan}. Recall that the presence of such a square in a ternary circular word is indicated by one of the following factors in its codeword: 
\begin{equation}
00, 1111, 01010, 011011011, 110110110, 11101110111.\label{list}
\end{equation}
It can be directly verified that no one of the codewords from the table contains a factor from the list~(\ref{list}). Since a word of length at most 21 can contain squares of period at most 10, all the codewords in (\ref{tab1}) are square-free. 
\end{proof}

\begin{Lem}\label{l_9-21}
There are square-free circular codewords of length 16, 19, and 20. No such codewords of length 14 or 17 exist.
\end{Lem}
\begin{proof}
A codeword built from a closed walk never contains $00$ or $1111$, while $01010$ [$011011011$, $110110110$, $11101110111$] in the codeword corresponds to $11$ [$222$ or $223$, $222$ or $322$, $333$, respectively] in the closed walk. So, 
\begin{itemize}
\item[(${*}{*}{*}$)] a circular codeword does not contain a factor from the list~(\ref{list}) if and only if the label of the corresponding closed walk has no proper factors $11$, $222$, $223$, $322$, and $333$.
\end{itemize}
Since any walk from (22)+(33) contains $223$ and any walk from (11)+(11)+(22) contains $11$, a closed walk cannot define a square-free circular codeword of length 14. A similar check for the sets (22)+(1213) and (11)+(1323) shows that there exist no square-free circular codeword of length 17. On the other hand, the walks $(122122)\in(11)$+(22)+(22), $(123313)\in(1213)$+(33), $(133133)\in(11)$+(33)+(33) define square-free circular codewords of length 16, 19, and 20, respectively. 
\end{proof}

In the proof of Lemma \ref{cycl} we have shown that there are no square-free circular codewords of length 9 or 10. Hence Lemmas \ref{cycl} and \ref{l_9-21} give us

\begin{Cor}\label{c_9-21}
Theorem \ref{sf} holds for $9\le l\le21$.
\end{Cor}

\subsection{Constructing walks of given weight}\label{walks}

In order to prove Theorem \ref{sf}, it remains to find, for any $l\ge22$, a closed walk which defines a square-free circular codeword of length $l$. A condition which ensures that a closed walk generates such a codeword is given by the following lemma.

\begin{Lem} \label{l2}
A closed walk having
\begin{itemize}
\item[(a)] no factors $11$, $222$, $223$, $322$, $333$, and 
\item[(b)] no factors of length $2t{-}2$ with the even period $t\ge4$ and the root being the label of a closed walk, 
\end{itemize}
defines a square-free circular codeword.
\end{Lem}

\begin{proof}
Assume that the circular codeword $\mathbf{w}$ is defined by some closed walk. If the circular word coded by $\mathbf{w}$ contains a square, then $\mathbf{w}$ contains a codeword of some minimal square $u^2$. Then the circular word $(u)$ is square-free by Proposition~\ref{l1}. Note that $|u|\ge4$, because the circular codewords defined by closed walks avoid the codewords $00$ and $1111$ of shorter squares. By Remark~\ref{r1}, the codeword of $(u)$ is defined by a closed walk. Let $t$ be the length of this walk. Since the graph $K_{3{,}3}$ is bipartite, $t$ is even. If $t=2$, the codeword of $(u)$ contains one of the factors listed in~(\ref{list}); the statement of the lemma now follows from (${*}{*}{*}$). So, let $t\ge4$.

``Doubling'' the closed walk which defines the codeword of $(u)$, we obtain the walk defining the codeword of $(u^2)$. Deleting two neighbouring letters in this codeword, we get the codeword of the word $u^2$. This deletion corrupts one or two labels in the ``doubled'' closed walk. The remaining part of this walk is a factor of the closed walk defining $\mathbf{w}$. This part has the period $t$ and the length at least $2t{-}2$. The lemma is proved.
\end{proof}

Thus, in the rest of the proof we construct the closed walks satisfying the conditions of Lemma~\ref{l2}.
We use a closed walk $(121212)$ of weight 15 and four closed walks of weight 18: $(122133)\in(11)$+(22)+(33), (123123), (132132), and (131313). The key role in the remaining considerations is played by the morphism $h:\{a,b,c\}^*\to\{1,2,3\}^*$, defined by
$$
h(1)=122133,\ h(2)=123123,\ h(3)=132132.
$$
Any circular word of the form $(h(u))$ is obviously a closed walk. Moreover, the following lemma holds.

\begin{Lem} \label{l3}
If $u\in\{a,b,c\}^*$ is a square-free word, then the word $h(u)$ contains none of the factors mentioned in Lemma~\ref{l2}. Moreover, if one replaces any block in $h(u)$ by the word $131313$ or $121212$, or replaces any two adjacent blocks by the word $131313121212$, then the resulting word also has none of the factors mentioned in Lemma~\ref{l2}.
\end{Lem}

\begin{proof}
The distance between any two consecutive $1$'s in the word $h(u)$ equals 3. Hence, $h(u)$ has no factors from the list (a), and any even period $t\ge4$ of any factor of $h(u)$ is divisible by 6. Let us prove that $h(u)$ contains no $t$-periodic factors of length $2t{-}2$. 

Assume that such a factor $xyx$ exists, that is, $|x|=t{-}2$, $|y|=2$. Since $t$ is divisible by the length of the block, both occurrences of $x$ can be represented in the same form $v_1wv_2$, where $v_1$ is a suffix of a block, $v_2$ is a prefix of a block, and $w$ is a product of some blocks. So, $v_1wv_2yv_1wv_2$ is a factor of $h(u)$. Note that the length of a common prefix [suffix] of two different blocks is at most 2 [respectively, at most 1]. Since $|y|=2$, we have $|v_1|+|v_2|=4$ or $|v_1|+|v_2|=10$. In the first case the block $z=v_2yv_1$ is uniquely determined. Furthermore, if $|v_2|>2$, then $v_2$ determines the block $z$, thus implying that $h(u)$ contains the factor $wzwz$. Similarly, if $|v_1|>1$, then $v_1$ determines the block $z$, and $h(u)$ contains $zwzw$. Since $u$ is square-free, we get a contradiction. In the second case we have $z_2z_1=v_2yv_1$, where the blocks $z_1$ and $z_2$ are determined by $v_1$ and $v_2$, respectively. Hence, $h(u)$ contains the factor $z_1wz_2z_1wz_2$ which contradicts to the square-freeness of $u$. The first statement of the lemma is proved.

We prove the second statement for the word $131313$ (the two other cases are similar). Let $z$ denote the word obtained by replacing some block in $h(u)$ with $131313$. Obviously, none of the factors from the list (a) intersects $131313$, while the walks $1313$ and $3131$ are open so that their appearance as factors of $z$ is not restricted by the condition (b). Furthermore, the word $131313$ and any block have a common prefix of length at most 2 and a common suffix of length at most 1. Thus, the longest factor with the root $131313$ or $313131$ has the length at most 9 and is allowed by the condition (b). Concerning the other factors of $z$ of the form $xyx$, it is easy to see from the definition of $h$ that $x$ cannot contain $131$. Hence, $x$ intersects a prefix of $131313$ of length at most 2 or a suffix of $131313$ of length 1. Therefore, we can reproduce the argument of the previous paragraph to show that $z$ has no factors mentioned in the condition (b).
\end{proof}

\begin{Cor} \label{z3z2}
Suppose that $u\in\{a,b,c\}^*$ is a square-free word and $z_3$ [$z_2$, $z_{32}$] is the word obtained from $h(u)$ by replacing the last block with the word $131313$ [respectively, the last block with $121212$, two last blocks with $131313121212$]. Then the circular words $(z_3)$, $(z_2)$, and $(z_{32})$ considered as closed walks define square-free circular codewords. 
\end{Cor}

\begin{proof}
By Lemma~\ref{l3}, any conjugate of the word $z_3$ [$z_2$, $z_{32}$] has none of the factors mentioned in Lemma~\ref{l2}. Then the circular words $(z_3)$, $(z_2)$, and $(z_{32})$ also have no such factors. The result now follows from Lemma~\ref{l2}.
\end{proof}

\begin{Cor} \label{18n}
Theorem \ref{sf} holds for $l=18n$ and $l=18n+15$ for any $n\ge1$.
\end{Cor}

\begin{proof}
Take an arbitrary square-free word $u\in\{a,b,c\}^*$ of length $n$. The walks $(z_3)$ and $(z_2)$ described in Corollary \ref{z3z2} define square-free circular codewords of the required length.
\end{proof}

\begin{proof}[Proof of Theorem \ref{sf}: remaining cases]
In order to get square-free circular codewords of length $18n{+}m$ for any $m\le17$ we slightly modify the circular words of type $(z_3)$, $(z_2)$, or $(z_{32})$ using short closed walks that are already known to define such codewords. Note that due to symmetry we can choose the letter which occupies any single fixed position in $u$.

Choose the block preceding $131313$ in $(z_3)$ [preceding $121212$ in $(z_2)$] to be $h(a)$ and consider the closed walk $(133133)$ of weight 20. If we replace the chosen block by $133133$, the resulting closed walk will satisfy the conditions of Lemma~\ref{l2}. Indeed, the block preceding $133133$ is not equal to $h(a)$, because $u$ is square-free, and hence has the common suffix of length $\le1$ with $133133$. So, we get no forbidden factor whose root is a conjugate of $133133$. Suppose that a factor of the form $xyx$ with $|y|=2$ intersects $133133$. It is easy to see that $x$ contains neither $133133$ nor the factor $131$ [respectively, $121$] which follows $133133$. Hence, $xyx$ ends with some prefix of $133133$. The length of this prefix is at most 2, because neither block begins with 133. Then $x$ begins with a suffix of length at least 2 of some block; such a prefix uniquely determines a block, so we discover a square in $u$ as in the proof of Lemma~\ref{l3} to obtain a contradiction. Thus, we obtained square-free circular codewords of length $18n{+}2$ and $18n{-}1$ for any $n\ge2$.

The argument in the remaining cases is similar or even easier than the one above. So, we just draw a table showing, for each $m$, what replacement in which circular word should be made in order to get a walk defining a square-free circular codeword of length $18n+m$ (the minimum possible $n$ is given in the second column of this table).\\[10pt]
\centerline{
\begin{tabular}{|r|rl|c|rcl|}
\hline
$m$&\multicolumn{2}{c|}{start with}&additional requirement&\multicolumn{3}{c|}{replace}\\
\hline
1&$(z_{32})$&$n{\ge}2$&no&131313&${\mapsto}$&13121213\\
2&$(z_3)$&$n{\ge}2$&block $s=122133$ precedes $131313$&$s$&${\mapsto}$&133133\\
3&$(z_2)$&$n{\ge}2$&block $s=122133$ precedes $121212$&$s$&${\mapsto}$&12212332\\
4&$(z_3)$&$n{\ge}1$&no&131313&${\mapsto}$&13121213\\
5&$(z_2)$&$n{\ge}2$&block $s=123123$ precedes $121212$&$s$&${\mapsto}$&12332133\\
6&$(z_3)$&$n{\ge}2$&block $s=122133$ precedes $131313$&$s$&${\mapsto}$&12212332\\
7&$(z_{32})$&$n{=}2$&no&$131313$&${\mapsto}$&1221312213\\
&&$n{\ge}3$&block $s=122133$ precedes $131313$&$s$&${\mapsto}$&1221312213\\
8&$(z_3)$&$n{\ge}2$&block $s=123123$ precedes $131313$&$s$&${\mapsto}$&12332133\\
9&$(z_2)$&$n{\ge}2$&block $s=122133$ precedes $121212$&$s$&${\mapsto}$&1221312323\\
10&$(z_3)$&$n{\ge}2$&block $s=122133$ precedes $131313$,&$s$&${\mapsto}$&1221312213\\
&&&block $132132$ does not follow $131313$&&&\\
11&$(z_2)$&$n{\ge}2$&block $s=123123$ precedes $121212$&$s$&${\mapsto}$&1233212332\\
12&$(z_3)$&$n{\ge}2$&block $s=122133$ precedes $131313$&$s$&${\mapsto}$&1221312323\\
13&$(z_2)$&$n{\ge}1$&block $s=122133$ follows $121212$&$s$&${\mapsto}$&122122\\
14&$(z_3)$&$n{\ge}2$&block $s=123123$ precedes $131313$&$s$&${\mapsto}$&1233212332\\
16&$(z_3)$&$n{\ge}1$&block $s=122133$ follows $131313$&$s$&${\mapsto}$&122122\\
17&$(z_2)$&$n{\ge}1$&block $s=122133$ precedes $121212$&$s$&${\mapsto}$&133133\\
\hline
\end{tabular} }\\[12pt]
Since Theorem~\ref{sf} is already proved for the cases $l=18n$, and $l=18n+15$ (Corollary~\ref{18n}), this table provides the proof for all $l\ge33$ and also for $l=22$ and $l=31$. To cover the cases $l=24,26,28,30,32$ we just take the ``replacing blocks'' from the corresponding rows of the table. Finally, the walks $12213132$, $12321323$, $1221221213$, and $1221221323$ provide square-free circular codewords of lengths $23,25,27$, and 29, respectively. In view of Corollaries~\ref{c_le8} and~\ref{c_9-21}, the proof of Theorem~\ref{sf} is finished.
\end{proof}

\section{Concluding remarks}

With the aid of the construction defined in Section~\ref{walks}, it is easy to show that the number of ternary square-free circular words grows exponentially with the length. Indeed, it is well known that there are exponentially many ternary square-free words of length $n$; each of these words can be used to build at least one ternary square-free circular word of length $18n{+}m$ for any $m=0,\ldots,17$. From computer experiments, we learn that the growth rate for the set of ternary square-free circular words is approximately $1.3$. 

This growth rate obviously cannot exceed the growth rate for the set of ternary square-free ordinary words. For the latter one, a nearly exact value is now known \cite{Sh}. It is $1.30176...$, so it is quite intriguing whether the two considered growth rates coincide.

Our proof can also be used to list the lengths for which the ternary square-free circular word is unique up to isomorphism: these are $1,2,3,4,6,8,11,12,13,15,16,21$; thus the smallest length for which there exist two non-isomorphic ternary square-free circular words is 18. This result can be interesting in terms of ``unique colourability'' of graphs.

\end{document}